\documentclass[11pt]{article}%
\usepackage{amsmath}
\usepackage{amsfonts}
\usepackage{amssymb}
\usepackage{amsthm}
\usepackage{graphicx}
\usepackage{fullpage}%
\providecommand{\U}[1]{\protect\rule{.1in}{.1in}}
\newtheorem{theorem}{Theorem}

\newtheorem{corollary}[theorem]{Corollary}

\newtheorem{definition}[theorem]{Definition}

\newtheorem{lemma}[theorem]{Lemma}

\newtheorem{proposition}[theorem]{Proposition}

\begin{document}

\title{Generalizing and Derandomizing Gurvits's Approximation Algorithm for the Permanent}
\author{Scott Aaronson\thanks{MIT. \ Email: aaronson@csail.mit.edu. \ \ This material
is based upon work supported by the National Science Foundation under Grant
No. 0844626. \ Also supported by a DARPA YFA grant, an NSF STC grant, a TIBCO
Chair, a Sloan Fellowship, and an Alan T.\ Waterman Award.}
\and Travis Hance\thanks{MIT. \ Email: tjhance7@gmail.com}}
\date{}
\maketitle

\begin{abstract}
Around 2002, Leonid Gurvits gave a striking randomized algorithm to
approximate the permanent of an $n\times n$\ matrix $A$. \ The algorithm runs
in $O\left(  n^{2}/\varepsilon^{2}\right)  $\ time, and approximates
$\operatorname*{Per}\left(  A\right)  $\ to within $\pm\varepsilon\left\Vert
A\right\Vert ^{n}$ additive error. \ A major advantage of Gurvits's algorithm
is that it works for \textit{arbitrary} matrices, not just for nonnegative
matrices. \ This makes it highly relevant to \textit{quantum optics}, where
the permanents of bounded-norm complex matrices play a central role. \ Indeed,
the existence of Gurvits's algorithm is why, in their recent work on the
\textit{hardness} of quantum optics, Aaronson and Arkhipov (AA) had to talk
about sampling problems rather than estimation problems.

In this paper, we improve Gurvits's algorithm in two ways. \ First, using an
idea from quantum optics, we generalize the algorithm so that it yields a
better approximation when the matrix $A$ has either repeated rows or repeated
columns. \ Translating back to quantum optics, this lets us classically
estimate the probability of \textit{any} outcome of an AA-type
experiment---even an outcome involving multiple photons \textquotedblleft
bunched\textquotedblright\ in the same mode---at least as well as that
probability can be estimated by the experiment itself. \ (This does not, of
course, let us solve the AA sampling problem.) \ It also yields a general
upper bound on the probabilities of \textquotedblleft
bunched\textquotedblright\ outcomes, which resolves a conjecture of Gurvits
and might be of independent physical interest.

Second, we use $\varepsilon$-biased sets to derandomize Gurvits's algorithm,
in the special case where the matrix $A$ is nonnegative. \ More interestingly,
we generalize the notion of $\varepsilon$-biased sets to the complex numbers,
construct \textquotedblleft complex $\varepsilon$-biased
sets,\textquotedblright\ then use those sets to derandomize even our
generalization of Gurvits's algorithm to the multirow/multicolumn case (again
for nonnegative $A$). \ Whether Gurvits's algorithm can be derandomized for
general $A$ remains an outstanding problem.

\end{abstract}

\section{Introduction\label{INTRO}}

The \textit{permanent} of an $n\times n$\ matrix has played a major role in
combinatorics and theoretical computer science for decades. \ Today this
function also plays an increasing role in quantum computing, because of the
fact (pointed out by Caianiello \cite{caianiello} and Troyansky and Tishby
\cite{troyanskytishby}\ among others) that the transition amplitudes for $n$
identical, non-interacting bosons are given by $n\times n$\ permanents.

Recall that the permanent of an $n\times n$ matrix $A=\left(  a_{ij}\right)
$, with entries over any field, is defined as follows:%
\begin{equation}
\operatorname*{Per}\left(  A\right)  :=\sum_{\sigma\in S_{n}}\prod_{i=1}%
^{n}a_{i,\sigma\left(  i\right)  }, \label{basic-permanent-def-eqn}%
\end{equation}
where $S_{n}$ denotes the set of permutations of $\left\{  1,\dotsc,n\right\}
$. \ A great deal is known about the complexity of computing
$\operatorname*{Per}\left(  A\right)  $, given $A$ as input. \ Most famously,
Valiant \cite{valiant} showed in 1979 that computing $\operatorname*{Per}%
\left(  A\right)  $ \textit{exactly} is $\mathsf{\#P}$-complete, even if $A$
is a matrix of nonnegative integers (or indeed, if all of its entries are
either $0$ or $1$). \ On the other hand, in a celebrated application of the
Markov-Chain Monte-Carlo method, Jerrum, Sinclair, and Vigoda \cite{jsv} gave
a randomized algorithm to \textit{approximate} $\operatorname*{Per}\left(
A\right)  $\ to multiplicative error $\varepsilon$, in $\operatorname*{poly}%
\left(  n,1/\varepsilon\right)  $\ time, in the special case where $A$ is
nonnegative. \ When $A$ is an \textit{arbitrary} matrix (which could include
negative or complex entries), it is easy to show that even approximating
$\operatorname*{Per}\left(  A\right)  $\ to polynomial multiplicative error
remains a $\mathsf{\#P}$-hard problem; see Aaronson \cite{aar:per}\ for example.

We can, on the other hand, hope to estimate $\operatorname*{Per}\left(
A\right)  $\ to within nontrivial \textit{additive} error for arbitrary
$A\in\mathbb{C}^{n\times n}$. \ And that is exactly what a simple randomized
algorithm due to Gurvits \cite{gurvits:alg} does: it approximates
$\operatorname*{Per}\left(  A\right)  $\ to within additive error
$\pm\varepsilon\left\Vert A\right\Vert ^{n}$\ in $O\left(  n^{2}%
/\varepsilon^{2}\right)  $\ time. \ Here $\left\Vert A\right\Vert
$\ represents the largest\textit{ singular value} of $A$.\ \ The appearance of
this linear-algebraic quantity when discussing the permanent becomes a little
less surprising once one knows the inequality%
\begin{equation}
\left\vert \operatorname*{Per}\left(  A\right)  \right\vert \leq\left\Vert
A\right\Vert ^{n}. \label{peran}%
\end{equation}
See Section \ref{GURVSEC} for Gurvits's algorithm, the proof of its
correctness, and the closely-related proof of the inequality (\ref{peran}).

In particular, if $U$ is a \textit{unitary} matrix, or a submatrix of a
unitary matrix, then $\left\Vert U\right\Vert \leq1$, so Gurvits's algorithm
approximates $\operatorname*{Per}\left(  U\right)  $\ to within $\pm
\varepsilon$\ additive error in $O\left(  n^{2}/\varepsilon^{2}\right)
$\ time. \ Equivalently, the algorithm lets us detect when
$\operatorname*{Per}\left(  U\right)  $\ is \textquotedblleft anomalously
close to $1$\textquotedblright\ (note that, if $U$ is a \textit{random}
unitary matrix, then $\operatorname*{Per}\left(  U\right)  $\ will usually be
exponentially small,\ and $0$ will accordingly be a good additive estimate to
$\operatorname*{Per}\left(  U\right)  $). \ This observation makes Gurvits's
algorithm highly relevant to the field of quantum optics, where subunitary
matrices play a central role.\ 

\subsection{Our Results\label{RESULTS}}

In this paper, motivated by the problem of simulating quantum optics on a
classical computer, we study how far Gurvits's algorithm can be improved. \ We
present two sets of results on this question.

First, we generalize Gurvits's algorithm, as well as the inequality
(\ref{peran}), to perform better on matrices with repeated rows or repeated
columns. \ Let $B\in\mathbb{C}^{n\times k}$ be an $n\times k$ complex matrix,
and let $s_{1},\dotsc,s_{k}$ be positive integers summing to $n$. \ Also, let
$A\in\mathbb{C}^{n\times n}$ be an $n\times n$ matrix consisting of $s_{i}$
copies of the $i^{th}$\ column of $B$, for all $i\in\left[  k\right]  $.
\ Then we give a randomized algorithm that takes $O\left(  nk/\varepsilon
^{2}\right)  $\ time, and that estimates $\operatorname*{Per}\left(  A\right)
$ to within an additive error%
\begin{equation}
\pm\varepsilon\cdot\frac{s_{1}!\dotsm s_{k}!}{\sqrt{s_{1}^{s_{1}}\dotsm
s_{k}^{s_{k}}}}\left\Vert B\right\Vert ^{n}. \label{pest0}%
\end{equation}
Interestingly, we obtain this generalization by interpreting certain formal
variables $x_{1},\dotsc,x_{k}$\ arising in Gurvits's algorithm as what a
physicist would call \textquotedblleft bosonic creation
operators,\textquotedblright\ and then doing what would be natural for
creation operators (e.g., replacing each $x_{i}$\ by $\sqrt{s_{i}}x_{i}$).

We also prove the general inequality%
\begin{equation}
\left\vert \operatorname*{Per}\left(  A\right)  \right\vert \leq\frac
{s_{1}!\dotsm s_{k}!}{\sqrt{s_{1}^{s_{1}}\dotsm s_{k}^{s_{k}}}}\left\Vert
B\right\Vert ^{n}. \label{pub0}%
\end{equation}
Notice that both of these results reduce to Gurvits's in the special case
$k=n$\ and $s_{1}=\dotsm=s_{n}=1$, but in general they can be much better (for
example, if we have information about $\left\Vert B\right\Vert $\ but not
$\left\Vert A\right\Vert $). \ We discuss the quantum-optics motivation for
these improvements in Section \ref{IMP}.

Second, we show that both Gurvits's algorithm and our generalization of it can
be \textit{derandomized}, in the special case where all matrix entries are
nonnegative. \ That is, for nonnegative $n\times n$\ matrices $A\in
\mathbb{R}_{\geq0}^{n\times n}$, we show that we can approximate
$\operatorname*{Per}\left(  A\right)  $ to within additive error
$\pm\varepsilon\cdot\Vert A\Vert^{n}$, deterministically and in
$\operatorname*{poly}\left(  n,1/\varepsilon\right)  $ time. \ If
$A\in\mathbb{R}_{\geq0}^{n\times n}$ is obtained from an $n\times k$\ matrix
$B\in\mathbb{R}_{\geq0}^{n\times k}$ as described above, then we also show
that we can approximate $\operatorname*{Per}\left(  A\right)  $\ to within
additive error%
\begin{equation}
\pm\varepsilon\cdot\frac{s_{1}!\dotsm s_{k}!}{\sqrt{s_{1}^{s_{1}}\dotsm
s_{k}^{s_{k}}}}\Vert B\Vert^{n},
\end{equation}
deterministically and in $\operatorname*{poly}\left(  n,1/\varepsilon\right)
$ time.

To derandomize Gurvits's original algorithm, the idea is simply to use
$\varepsilon$-biased sets. \ To derandomize our generalization of Gurvits's
algorithm is more interesting: there, we need to generalize the notion of
$\varepsilon$-biased sets to complex roots of unity, then explicitly construct
the \textquotedblleft complex $\varepsilon$-biased sets\textquotedblright%
\ that we need.

Let us compare our algorithm to a previous deterministic approximation
algorithm for the permanent. \ Gurvits and Samorodnitsky \cite{gurvitsky} gave
a deterministic algorithm for estimating the \textit{mixed discriminant} of a
positive semidefinite matrix up to a multiplicative error of $e^{n}$. \ The
permanent of a nonnegative matrix is a special case of this. \ Our algorithm
improves over Gurvits and Samorodnitsky's \cite{gurvitsky} when
$\operatorname*{Per}\left(  A\right)  $ is large, i.e., close to $\left\Vert
A\right\Vert ^{n}$ or to $\frac{s_{1}!\dotsm s_{k}!}{\sqrt{s_{1}^{s_{1}}\dotsm
s_{k}^{s_{k}}}}\Vert B\Vert^{n}$.

Of course, for nonnegative $A$, the algorithm of Jerrum, Sinclair, Vigoda
\cite{jsv} (JSV) gives a better approximation to $\operatorname*{Per}\left(
A\right)  $\ than any of the algorithms discussed above (including ours).
\ Crucially, though, no one seems to have any idea how to derandomize JSV.
\ One motivation for trying to derandomize Gurvits's algorithm is the hope
that it \textit{might} be a stepping-stone toward a derandomization of the
much more complicated JSV algorithm.

\subsection{Quantum Optics\label{OPTICS}}

This paper is about classical algorithms for a classical permanent-estimation
problem. \ However, the motivation and even one of the algorithmic ideas came
from quantum optics, so a brief discussion of the latter might be helpful.
\ For details of quantum optics from a computer-science perspective, see for
example Aaronson \cite{aar:per} or Aaronson and Arkhipov \cite{aark}.

In \textit{quantum (linear) optics}, one considers states involving identical,
non-interacting photons (or other bosonic particles), which can move from one
location (or \textquotedblleft mode\textquotedblright) to another, but are
never created or destroyed. \ In more detail, the photons are in a
superposition of basis states, each of the form $\left\vert s_{1},\dotsc
,s_{k}\right\rangle $, like so:%
\begin{equation}
\left\vert \Phi\right\rangle =\sum_{s_{1},\dotsc,s_{k}\geq0~:~s_{1}%
+\dotsb+s_{k}=n}\alpha_{s_{1},\dotsc,s_{k}}\left\vert s_{1},\dotsc
,s_{k}\right\rangle ,
\end{equation}
for some complex numbers $\alpha_{s_{1},\dotsc,s_{k}}$\ satisfying%
\begin{equation}
\sum_{s_{1},\dotsc,s_{k}\geq0~:~s_{1}+\dotsb+s_{k}=n}\left\vert \alpha
_{s_{1},\dotsc,s_{k}}\right\vert ^{2}=1.
\end{equation}
Here $s_{i}$\ is a nonnegative integer representing the number of photons in
the $i^{th}$\ mode, and $s_{1}+\dotsb+s_{k}=n$\ is the total number of
photons. \ To modify the state $\left\vert \Phi\right\rangle $, one applies a
network of \textquotedblleft beamsplitters\textquotedblright\ and other
optical elements, which induces some $k\times k$ unitary transformation $U$
acting on the $k$\ modes. \ Via homomorphism, the $k\times k$\ unitary $U$
gives rise to a $\binom{k+n-1}{n}\times\binom{k+n-1}{n}$\ unitary
$\varphi\left(  U\right)  $\ acting on the $n$-photon state $\left\vert
\Phi\right\rangle $. \ Crucially, each entry of the \textquotedblleft
large\textquotedblright\ unitary $\varphi\left(  U\right)  $ is defined in
terms of the \textit{permanent} of a matrix formed from entries of the
\textquotedblleft small\textquotedblright\ unitary $U$. \ The formula is as
follows:%
\begin{equation}
\left\langle s_{1},\dotsc,s_{k}|\varphi\left(  U\right)  |t_{1},\dotsc
,t_{k}\right\rangle =\frac{\operatorname*{Per}\left(  U_{s_{1},\dotsc
,s_{k},t_{1},\dotsc,t_{k}}\right)  }{\sqrt{s_{1}!\cdots s_{k}!t_{1}!\cdots
t_{k}!}}. \label{lo}%
\end{equation}
Here $U_{s_{1},\dotsc,s_{k},t_{1},\dotsc,t_{k}}$\ is the $n\times n$\ matrix
formed from $U$ by taking $s_{i}$\ copies of the $i^{th}$\ row of $U$ for all
$i\in\left[  k\right]  $, and $t_{j}$\ copies of the $j^{th}$\ column of $U$
for all $j\in\left[  k\right]  $. \ It can be checked \cite{aark,aar:per}---it
is not obvious!---that $\varphi$\ is a homomorphism, and that $\varphi\left(
U\right)  $\ is unitary for all $U$. \ (Indeed, the \textquotedblleft
reason\textquotedblright\ for the scaling term $\sqrt{s_{1}!\cdots s_{k}%
!t_{1}!\cdots t_{k}!}$\ is to ensure these properties.)

Thus, in quantum optics, \textit{calculating amplitudes reduces to calculating
permanents of the above form}. \ By the standard rules of quantum mechanics,
the probabilities of measurement outcomes can then be obtained by taking the
absolute squares of the amplitudes.

Intuitively, the reason why the permanent arises here is simply that (by
assumption) the $n$\ photons are identical, and therefore we need to sum over
all $n!$ possible permutations of photons, each of which contributes a term to
the final amplitude.\ \ Indeed, even in a classical situation, with $n$
indistinguishable balls each thrown independently into one of $k$ bins, the
probability of a particular outcome (for example: $2$ balls landing in the
first bin, $0$ balls landing in the second bin, etc.) could be expressed in
terms of the permanent of a matrix $A$ of transition probabilities. \ The
difference is that, in the classical case, this $A$\ would be a
\textit{nonnegative} matrix. \ And therefore, $\operatorname*{Per}\left(
A\right)  $\ could be estimated in randomized polynomial time using the JSV
algorithm \cite{jsv}. \ In the quantum case, by contrast, we want to estimate
$\left\vert \operatorname*{Per}\left(  A\right)  \right\vert ^{2}$\ for a
matrix $A$ of \textit{complex }numbers---and here is it known that
multiplicative estimation is already a $\mathsf{\#P}$-hard problem
\cite{aark}. \ That is why we need to settle for \textit{additive}
approximation---or equivalently, for approximation in the special case that
$\left\vert \operatorname*{Per}\left(  A\right)  \right\vert ^{2}$\ happens to
be \textquotedblleft anomalously large\textquotedblright\ (i.e.,
non-negligible compared to the general upper bound that we prove on
$\left\vert \operatorname*{Per}\left(  A\right)  \right\vert ^{2}$).

\subsection{Implications of Our Results for Quantum Optics\label{IMP}}

With that brief tour of quantum optics out of the way, we can now explain the
implications for quantum optics of our first set of results (the ones
generalizing Gurvits's algorithm to the multirow/multicolumn case).

Consider the \textquotedblleft standard initial state\textquotedblright%
\ $\left\vert 1,\ldots,1,0,\ldots,0\right\rangle $, which consists of $n$
identical photons, each in a separate mode, and the remaining $m-n$\ modes
unoccupied. \ Suppose we pass the $n$ photons through a network of
beamsplitters, then measure the numbers of photons in each of $k$ output
modes. \ Given nonnegative integers $s_{1},\dotsc,s_{k}$ summing to $n$, let
$p$ be the probability that exactly $s_{1}$\ photons are found in the
first\ output mode, exactly $s_{2}$ are found in the second output mode, and
so on up to $s_{k}$. \ Then by (\ref{lo}), we have%
\begin{equation}
p=\frac{\left\vert \operatorname*{Per}\left(  A\right)  \right\vert ^{2}%
}{s_{1}!\cdots s_{k}!},
\end{equation}
where $A$ is the $n\times n$\ matrix formed by taking $s_{i}$ copies of the
$i^{th}$ row of the mode-mixing unitary $U$, for all $i\in\left[  k\right]  $.
\ From this, together with the inequality (\ref{pub0}), we have the upper
bound%
\begin{equation}
p\leq\frac{s_{1}!\dotsm s_{k}!}{s_{1}^{s_{1}}\dotsm s_{k}^{s_{k}}},
\label{pub}%
\end{equation}
and for the associated amplitude $\alpha$,%
\begin{equation}
\left\vert \alpha\right\vert \leq\sqrt{\frac{s_{1}!\dotsm s_{k}!}{s_{1}%
^{s_{1}}\dotsm s_{k}^{s_{k}}}}.
\end{equation}

As a sample application of (\ref{pub}), suppose $s_{1}=n$\ and $s_{2}%
=\cdots=s_{k}=0$. \ Then%
\begin{equation}
p\leq\frac{n!}{n^{n}}\approx\frac{1}{e^{n}}.
\end{equation}
This says that, regardless of what unitary transformation $U$ we apply, we can
never get $n$ identical photons, initially in $n$ separate modes, to
\textquotedblleft congregate\textquotedblright\ into a single one of those
modes with more than $\sim n/e^{n}$ probability (the factor of $n$ arising
from taking the union bound over the modes). \ Or in other words, for $n\geq
3$, there is no direct counterpart to the \textit{Hong-Ou-Mandel} dip
\cite{hom}, the famous effect that causes $n=2$ photons initially in separate
modes to congregate into the same mode with probability $1$.\footnote{When
$n=2$, we get $p\leq1/2$, which corresponds exactly to the Hong-Ou-Mandel dip:
the $2$ photons have probability $1/2$ of both being found in the first mode,
and probability $1/2$ of both being found in the second mode.}

Let us remark that the bound (\ref{pub}) is tight. \ To saturate it, let the
$n\times n$\ unitary $U$\ be block-diagonal with $k$ blocks, of sizes
$s_{1}\times s_{1},\ldots,s_{k}\times s_{k}$. \ Also, let the $i^{th}$\ block
consist of the $s_{i}$-dimensional Quantum Fourier Transform,\ or any other
$s_{i}\times s_{i}$\ unitary matrix whose top row equals $\left(
1/\sqrt{s_{i}},\ldots,1/\sqrt{s_{i}}\right)  $. \ Then one can calculate that
the probability of observing $s_{1}$\ photons in the first mode of the first
block, $s_{2}$\ photons in the first mode of the second block, and so on is
precisely%
\begin{equation}
\left\vert \left\langle 1,\ldots,1|\varphi\left(  U\right)  |s_{1}%
,0,\dotsc,0,s_{2},0,\dotsc,0,\dotsc\right\rangle \right\vert ^{2}=\frac
{s_{1}!\dotsm s_{k}!}{s_{1}^{s_{1}}\dotsm s_{k}^{s_{k}}}.
\end{equation}

Here is another implication of our results for quantum optics. \ Given a
description of the unitary transformation $U$ applied by the beamsplitter
network, the bound (\ref{pest0}) implies the existence of a randomized
algorithm, taking $O\left(  nk/\varepsilon^{2}\right)  $\ time, that estimates
the amplitude $\alpha=\frac{\operatorname*{Per}\left(  A\right)  }{\sqrt
{s_{1}!\dotsm s_{k}!}}$\ to within an additive error of%
\begin{align}
\varepsilon\cdot\sqrt{\frac{s_{1}!\dotsm s_{k}!}{s_{1}^{s_{1}}\dotsm
s_{k}^{s_{k}}}}\left\Vert B\right\Vert ^{n}  &  \leq\varepsilon\cdot
\sqrt{\frac{s_{1}!\dotsm s_{k}!}{s_{1}^{s_{1}}\dotsm s_{k}^{s_{k}}}}\\
&  \leq\varepsilon,
\end{align}
and likewise estimates the probability $p=\left\vert \alpha\right\vert ^{2}$
to within an additive error of%
\begin{align}
\varepsilon\cdot\left\vert \alpha\right\vert \sqrt{\frac{s_{1}!\dotsm s_{k}%
!}{s_{1}^{s_{1}}\dotsm s_{k}^{s_{k}}}}\left\Vert B\right\Vert ^{n}  &
\leq\varepsilon\cdot\frac{s_{1}!\dotsm s_{k}!}{s_{1}^{s_{1}}\dotsm
s_{k}^{s_{k}}}\left\Vert B\right\Vert ^{n}\label{pest}\\
&  \leq\varepsilon\cdot\frac{s_{1}!\dotsm s_{k}!}{s_{1}^{s_{1}}\dotsm
s_{k}^{s_{k}}}\\
&  \leq\varepsilon.
\end{align}
Observe that the above guarantees match Gurvits's in the special case that all
$s_{i}$'s are equal to $1$, but they become \textit{better} than Gurvits's
when\ $\sum_{i}\left\vert s_{i}-1\right\vert $ is large. \ In the latter case,
(\ref{pub}) says that the probability $p$ is exponentially small, but
(\ref{pest}) says that we can \textit{nevertheless} get a decent estimate of
$p$.

\section{Gurvits's Randomized Algorithm for Permanent
Estimation\label{GURVSEC}}

The starting point for Gurvits's algorithm is the following well-known formula
for the permanent of an $n\times n$ matrix, called \textit{Ryser's formula}
\cite{ryser}:
\begin{equation}
\operatorname*{Per}(A)=\sum_{x=\left(  x_{1},\dotsc,x_{n}\right)
\in\{0,1\}^{n}}(-1)^{x_{1}+\dotsb+x_{n}}\prod_{i=1}^{n}\left(  a_{i,1}%
x_{1}+\dotsb+a_{i,n}x_{n}\right)  .
\end{equation}
Ryser's formula leads to an $O\left(  2^{n}n^{2}\right)  $ algorithm for
computing the permanent exactly, which can be improved to $O\left(
2^{n}n\right)  $ by cleverly iterating through the $x$'s in Gray code order.
\ This remains the fastest-known exact algorithm for the
permanent.\footnote{Recently Bj\"{o}rklund \cite{bjorklund} discovered a
slightly faster algorithm, taking $2^{n-\Omega\left(  \sqrt{n/\log n}\right)
}$ time, for computing the permanent of a matrix of $\operatorname*{poly}%
\left(  n\right)  $-bit integers.}

To present Gurvits's randomized algorithm, we will use a similar formula due
to Glynn \cite{glynn}, which pulls from the domain $\left\{  -1,1\right\}
^{n}$ rather than $\left\{  0,1\right\}  ^{n}$. \ We state Glynn's formula in
terms of the expectation of a random variable. \ For a given $x\in\left\{
-1,1\right\}  ^{n}$, define the \textit{Glynn estimator} of an $n\times n$
matrix $A$ as
\begin{equation}
\operatorname*{Gly}\nolimits_{x}(A):=x_{1}\dotsm x_{n}\prod_{i=1}^{n}\left(
a_{i,1}x_{1}+\dotsb+a_{i,n}x_{n}\right)  .
\end{equation}
Then we have%
\begin{equation}
\operatorname*{Per}\left(  A\right)  =\operatorname*{E}_{x\in\{-1,1\}^{n}%
}\left[  {\operatorname*{Gly}\nolimits_{x}(A)}\right]  .
\end{equation}
To see why, we just need to expand the product:
\begin{equation}
\operatorname*{E}_{x\in\{-1,1\}^{n}}\left[  {\operatorname*{Gly}\nolimits_{x}%
}\left(  {A}\right)  \right]  =\sum_{\sigma_{1},\dotsc,\sigma_{n}\in\lbrack
n]}a_{1,\sigma_{1}}\dotsm a_{n,\sigma_{n}}\operatorname*{E}_{x\in\{-1,1\}^{n}%
}\left[  {(x_{1}\dotsm x_{n})(x_{\sigma_{1}}\dotsm x_{\sigma_{n}})}\right]  .
\label{expected-gly-expansion}%
\end{equation}
Then, note that $\operatorname*{E}_{x\in\{-1,1\}^{n}}[(x_{1}\dotsm
x_{n})(x_{\sigma_{1}}\dotsm x_{\sigma_{n}})]$ is $1$ if the map $i\mapsto
\sigma_{i}$ is a permutation, and $0$ otherwise. \ Hence the above sum is
simply $\operatorname*{Per}(A)$.

As a special case of a more general algorithm for \textit{mixed discriminants}%
, Gurvits \cite{gurvits:alg} gave a polynomial-time randomized sampling
algorithm for the permanent. \ His algorithm is simply the following: for some
$m=O\left(  1/\varepsilon^{2}\right)  $, first choose $n$-bit strings\ $x_{1}%
,\ldots,x_{m}\in\left\{  -1,1\right\}  ^{n}$ uniformly and independently at
random. \ Then compute ${\operatorname*{Gly}\nolimits_{x_{j}}}\left(
{A}\right)  $\ for all $j\in\left[  m\right]  $, and output%
\begin{equation}
\frac{{\operatorname*{Gly}\nolimits_{x_{1}}}\left(  {A}\right)  +\cdots
+{\operatorname*{Gly}\nolimits_{x_{m}}}\left(  {A}\right)  }{m}%
\end{equation}
as the estimate for $\operatorname*{Per}\left(  A\right)  $.

Since each ${\operatorname*{Gly}\nolimits_{x_{i}}}\left(  {A}\right)  $\ can
be computed in $O\left(  n^{2}\right)  $ time, the algorithm clearly takes
$O\left(  n^{2}/\varepsilon^{2}\right)  $\ time. \ The algorithm's correctness
follows from a single (important) fact. \ Recall that $\left\Vert A\right\Vert
$\ denotes the largest singular value of $A$, or equivalently,%
\begin{equation}
\left\Vert A\right\Vert :=\sup_{x\neq0}\frac{\left\Vert Ax\right\Vert
}{\left\Vert x\right\Vert }%
\end{equation}
where $\left\Vert x\right\Vert $\ denotes the $2$-norm of the vector $x$.

\begin{proposition}
\label{gly-bound-prop}For any $n\times n$ complex matrix $A$, we have
$\left\vert \operatorname*{Gly}_{x}\left(  A\right)  \right\vert
\leq\left\Vert A\right\Vert ^{n}$.
\end{proposition}

\begin{proof}
For any $x\in\{1,-1\}^{n}$, we have by the arithmetic-geometric mean
inequality,
\begin{align}
\left\vert \operatorname*{Gly}\nolimits_{x}(A)\right\vert  &  \leq\left(
\sqrt{\frac{\sum_{i=1}^{n}\left\vert a_{i,1}x_{1}+\dotsb+a_{i,n}%
x_{n}\right\vert ^{2}}{n}}\right)  ^{n}\\
&  =\left(  \frac{\left\Vert Ax\right\Vert }{\left\Vert x\right\Vert }\right)
^{n}\\
&  \leq\left\Vert A\right\Vert ^{n}.
\end{align}
\end{proof}

To summarize, $\operatorname*{Gly}_{x}\left(  A\right)  $\ is a random
variable that is bounded by $\left\Vert A\right\Vert ^{n}$ in absolute value,
and whose expectation is $\operatorname*{Per}\left(  A\right)  $. \ This
implies that%
\begin{equation}
\left\vert \operatorname*{Per}\left(  A\right)  \right\vert \leq\left\Vert
A\right\Vert ^{n}.
\end{equation}
By a standard Chernoff bound, it also implies that we can estimate
$\operatorname*{Per}\left(  A\right)  $\ to additive error $\pm\varepsilon
\left\Vert A\right\Vert ^{n}$, with high probability, by taking the empirical
mean of $\operatorname*{Gly}_{x}\left(  A\right)  $\ for $O\left(
1/\varepsilon^{2}\right)  $\ random $x$'s.

\section{Generalized Gurvits Algorithm\label{GENGLYNN}}

We now give our generalization of Gurvits's algorithm to the
multirow/multicolumn case, as well as our generalized upper bound for the
permanent. \ As in Section \ref{RESULTS}, given nonnegative integers
$s_{1},\ldots,s_{k}$\ summing to $n$, we let $B\in\mathbb{C}^{n\times k}$\ be
an $n\times k$\ matrix, and let $A\in\mathbb{C}^{n\times n}$\ be the $n\times
n$\ matrix\ in which the $i^{th}$\ column of $B$ is repeated $s_{i}$\ times.
\ Also, let $\mathcal{R}[j]$ denote the set of the $j^{th}$ roots of unity,
and let $\mathcal{X}:=\mathcal{R}[s_{1}+1]\times\dotsb\times\mathcal{R}%
[s_{k}+1]$. \ Then for any $x\in\mathcal{X}$, we define the
\textit{generalized Glynn estimator} as follows. \ If $x=(x_{1},\dotsc,x_{k}%
)$, then let $y_{i}=\sqrt{s_{i}}x_{i}$ and
\begin{equation}
\operatorname*{GenGly}\nolimits_{x}(A):=\frac{s_{1}!\cdots s_{k}!}%
{s_{1}^{s_{1}}\dotsm s_{k}^{s_{k}}}\overline{y_{1}}^{s_{1}}\dotsm
\overline{y_{k}}^{s_{k}}\prod_{i=1}^{n}(y_{1}b_{i,1}+\dotsb+y_{k}b_{i,k})
\end{equation}
where $b_{i,j}$ denotes the $(i,j)$ entry of $B$. \ We will use this to
estimate $\operatorname*{Per}\left(  A\right)  $, as follows. \ First we will
show that%
\begin{equation}
\operatorname*{E}_{x}\left[  \operatorname*{GenGly}\nolimits_{x}\left(
A\right)  \right]  =\operatorname*{Per}\left(  A\right)  .
\end{equation}
Then we will show that%
\begin{equation}
\left\vert \operatorname*{GenGly}\nolimits_{x}\left(  A\right)  \right\vert
\leq\frac{s_{1}!\cdots s_{k}!}{\sqrt{s_{1}^{s_{1}}\dotsm s_{k}^{s_{k}}}%
}\left\Vert B\right\Vert ^{n}.
\end{equation}

The $\frac{s_{1}!\cdots s_{k}!}{s_{1}^{s_{1}}\dotsm s_{k}^{s_{k}}}$ factor in
the definition of $\operatorname*{GenGly}_{x}$ is a normalization factor that
we need for Lemma \ref{gengly-expected-lemma} below to hold. \ The value of
the normalization factor depends on our choice of the scale factors when
defining the $y_{i}$'s. \ Here, we set $y_{i}=\sqrt{s_{i}}x_{i}$ in order to
optimize the bound in Lemma \ref{gengly-bound-lemma} below.

\begin{lemma}
\label{gengly-expected-lemma}For the uniform distribution of $x$ over
$\mathcal{X}$, the expected value of $\operatorname*{GenGly}_{x}\left(
A\right)  $ is
\begin{equation}
\operatorname*{E}_{x\in\mathcal{X}}\left[  {\operatorname*{GenGly}%
\nolimits_{x}}\left(  {A}\right)  \right]  =\operatorname*{Per}\left(
A\right)  .
\end{equation}

\end{lemma}

\begin{proof}
We have
\begin{align}
&  \operatorname*{E}_{x\in\mathcal{X}}\left[  {\frac{s_{1}!\cdots s_{k}%
!}{s_{1}^{s_{1}}\dotsm s_{k}^{s_{k}}}\overline{y_{1}}^{s_{1}}\dotsm
\overline{y_{k}}^{s_{k}}\prod_{i=1}^{n}(y_{1}b_{i,1}+\dotsb+y_{k}b_{i,k}%
)}\right] \\
&  =\frac{s_{1}!\cdots s_{k}!}{s_{1}^{s_{1}}\dotsm s_{k}^{s_{k}}}\sum
_{\sigma_{1},\dotsc,\sigma_{n}\in\lbrack k]}b_{1,\sigma_{1}}\dotsm
b_{n,\sigma_{n}}\operatorname*{E}_{x\in\mathcal{X}}\left[  {(\overline{y_{1}%
}^{s_{1}}\dotsm\overline{y_{k}}^{s_{k}})(y_{\sigma_{1}}\dotsm y_{\sigma_{n}}%
)}\right]  .
\end{align}
Now notice that by symmetry over the roots of unity,%
\begin{equation}
\operatorname*{E}_{x\in\mathcal{X}}\left[  {(\overline{y}_{1}^{s_{1}}%
\cdots\overline{y}_{k}^{s_{k}})(y_{t_{1}}\cdots y_{t_{n}})}\right]  =0,
\end{equation}
unless the product inside happens to evaluate to $\left\vert y_{1}\right\vert
^{2s_{1}}\cdots\left\vert y_{k}\right\vert ^{2s_{k}}$, in which case we have
\begin{equation}
\operatorname*{E}_{x\in\mathcal{X}}\left[  {\left\vert y_{1}\right\vert
^{2s_{1}}\cdots\left\vert y_{k}\right\vert ^{2s_{k}}}\right]  =s_{1}^{s_{1}%
}\cdots s_{k}^{s_{k}}\operatorname*{E}_{x\in\mathcal{X}}\left[  {\left\vert
x_{1}\right\vert ^{2s_{1}}\cdots\left\vert x_{k}\right\vert ^{2s_{k}}}\right]
=s_{1}^{s_{1}}\cdots s_{k}^{s_{k}}.
\end{equation}
Therefore,
\begin{align}
\operatorname*{E}_{x\in\mathcal{X}}\left[  {\operatorname*{GenGly}%
\nolimits_{x}\left(  A\right)  }\right]   &  =s_{1}!\cdots s_{k}!\sum
_{t_{1},\ldots,t_{n}\in\left[  k\right]  ~:~\left\vert \left\{  i:t_{i}%
=j\right\}  \right\vert =s_{j}~\forall j\in\left[  k\right]  }\prod_{i=1}%
^{n}b_{i,t_{i}}\\
&  =\operatorname*{Per}\left(  A\right)  \label{divide-by-factorials-eqn}%
\end{align}
where (\ref{divide-by-factorials-eqn}) follows from the fact that each product
$\prod_{i=1}^{n}b_{i,t_{i}}$ appears $s_{1}!\dotsm s_{k}!$ times in equation
(\ref{basic-permanent-def-eqn}).
\end{proof}

\begin{lemma}
\label{gengly-bound-lemma}For all $x\in\mathcal{X}$,
\begin{equation}
\left\vert \operatorname*{GenGly}\nolimits_{x}\left(  A\right)  \right\vert
\leq\frac{s_{1}!\cdots s_{k}!}{\sqrt{s_{1}^{s_{1}}\cdots s_{k}^{s_{k}}}%
}\left\Vert B\right\Vert ^{n}.
\end{equation}

\end{lemma}

\begin{proof}
Let $y=\left(  y_{1},\ldots,y_{k}\right)  $. \ Then
\begin{align}
\left\Vert y\right\Vert _{2}  &  =\sqrt{\left\vert y_{1}\right\vert
^{2}+\cdots+\left\vert y_{k}\right\vert ^{2}}\\
&  =\sqrt{s_{1}+\cdots+s_{k}}\\
&  =\sqrt{n}.
\end{align}
Hence
\begin{equation}
\left\Vert By\right\Vert \leq\left\Vert B\right\Vert \left\Vert y\right\Vert
=\sqrt{n}\left\Vert B\right\Vert .
\end{equation}
So letting $\left(  By\right)  _{i}$ be the $i^{th}$\ entry of $By$, we have%
\begin{equation}
\left\vert \prod_{i=1}^{n}\left(  By\right)  _{i}\right\vert \leq\left(
\sqrt{\frac{n\left\Vert B\right\Vert ^{2}}{n}}\right)  ^{n}=\left\Vert
B\right\Vert ^{n}%
\end{equation}
by the arithmetic-geometric mean inequality. \ Therefore
\begin{align}
\left\vert \operatorname*{GenGly}\nolimits_{x}\left(  A\right)  \right\vert
&  =\left\vert \frac{s_{1}!\cdots s_{k}!}{s_{1}^{s_{1}}\dotsm s_{k}^{s_{k}}%
}\overline{y}_{1}^{s_{1}}\cdots\overline{y}_{k}^{s_{k}}\prod_{i=1}^{n}\left(
By\right)  _{i}\right\vert \\
&  \leq\frac{s_{1}!\cdots s_{k}!}{s_{1}^{s_{1}}\dotsm s_{k}^{s_{k}}}%
\sqrt{s_{1}^{s_{1}}\cdots s_{k}^{s_{k}}}\left\Vert B\right\Vert ^{n}\\
&  =\frac{s_{1}!\cdots s_{k}!}{\sqrt{s_{1}^{s_{1}}\dotsm s_{k}^{s_{k}}}%
}\left\Vert B\right\Vert ^{n}.
\end{align}
\end{proof}

To summarize, ${\operatorname*{GenGly}\nolimits_{x}}\left(  {A}\right)  $\ is
an unbiased estimator for $\operatorname*{Per}\left(  A\right)  $\ that is
upper-bounded by $\frac{s_{1}!\cdots s_{k}!}{\sqrt{s_{1}^{s_{1}}\dotsm
s_{k}^{s_{k}}}}\left\Vert B\right\Vert ^{n}$\ in absolute value. \ This
immediately implies that%
\begin{equation}
\left\vert \operatorname*{Per}\left(  A\right)  \right\vert \leq\frac
{s_{1}!\cdots s_{k}!}{\sqrt{s_{1}^{s_{1}}\dotsm s_{k}^{s_{k}}}}\left\Vert
B\right\Vert ^{n}.
\end{equation}
It also implies that there exists a randomized algorithm, taking $O\left(
nk/\varepsilon^{2}\right)  $ time, that approximates $\operatorname*{Per}%
\left(  A\right)  $\ to within additive error%
\begin{equation}
\pm\varepsilon\cdot\frac{s_{1}!\cdots s_{k}!}{\sqrt{s_{1}^{s_{1}}\dotsm
s_{k}^{s_{k}}}}\left\Vert B\right\Vert ^{n}%
\end{equation}
with high probability. \ Just like in Section \ref{GURVSEC}, that algorithm is
simply to choose $m=O\left(  1/\varepsilon^{2}\right)  $\ independent random
samples\ $x_{1},\ldots,x_{m}\in\mathcal{X}$, then output%
\begin{equation}
\frac{{\operatorname*{GenGly}\nolimits_{x_{1}}}\left(  {A}\right)
+\cdots+{\operatorname*{GenGly}\nolimits_{x_{m}}}\left(  {A}\right)  }{m}%
\end{equation}
as the estimate for $\operatorname*{Per}\left(  A\right)  $. \ The correctness
of this algorithm follows from a standard Chernoff bound.

As discussed in Section \ref{IMP}, in a linear-optics experiment with the
standard initial state, the above lets us estimate any output amplitude to
within additive error $\pm\varepsilon\sqrt{\frac{s_{1}!\dotsm s_{k}!}%
{s_{1}^{s_{1}}\dotsm s_{k}^{s_{k}}}}$, and any output probability to within
additive error $\pm\varepsilon\frac{s_{1}!\dotsm s_{k}!}{s_{1}^{s_{1}}\dotsm
s_{k}^{s_{k}}}$. \ This is at least as good an approximation as is provided by
running the experiment itself (and better, if some of the $s_{i}$'s are large).

\section{Derandomization\label{DERAND}}

We now discuss the derandomization of Gurvits's algorithm (and its
generalization) for nonnegative matrices.

\subsection{Derandomizing Gurvits\label{DEGURV}}

We start by showing how Gurvits's algorithm can be derandomized in the case of
computing the permanent of a $n\times n$ matrix with nonnegative real entries.
\ That is, we give a deterministic algorithm that estimates
$\operatorname*{Per}\left(  A\right)  $ to within $\pm\varepsilon\left\Vert
A\right\Vert ^{n}$ additive error with certainty.

To do this, we need a pseudorandom generator with a specific guarantee.

\begin{definition}
We say that a probability distribution $\mathcal{D}$ on $\{0,1\}^{n}$ is
$\varepsilon$\textbf{-biased} if for any nonzero vector $a\in\{0,1\}^{n}$, we
have
\begin{equation}
\left\vert \operatorname*{E}_{x\sim\mathcal{D}}\left[  {(-1)^{a\cdot x}%
}\right]  \right\vert \leq\varepsilon.
\end{equation}
Furthermore, we say that a probabilistic algorithm $G$ is $\varepsilon$-biased
if it outputs an $\varepsilon$-biased distribution.
\end{definition}

Note that if $G$ outputs the uniform distribution over $\{0,1\}^{n}$, then $G$
is $0$-biased. \ Unfortunately, this requires $n$ random bits, and so it takes
exponential time to loop over all possible states. \ We will need an
$\varepsilon$-biased generator with a much smaller seed. \ Such a generator
was constructed by Naor and Naor in \cite{naor-naor}, who showed the following.

\begin{theorem}
\label{eps-bias-thm}There exists an $\varepsilon$-biased generator which runs
in $\operatorname*{poly}(n,1/\varepsilon)$ time and uses%
\begin{equation}
O\left(  \log n+\log1/\varepsilon\right)
\end{equation}
bits of randomness.
\end{theorem}

This allows us to derandomize Gurvits's algorithm.

\begin{theorem}
\label{eps-bias-implies-approx-thm}If $\mathcal{D}$ is an $\varepsilon$-biased
distribution on $\{0,1\}^{n}$, then for any $n\times n$ matrix $A$ with
nonnegative real entries, we have
\begin{equation}
\left\vert \operatorname*{Per}(A)-\operatorname*{E}_{x\sim\mathcal{D}}\left[
{\operatorname*{Gly}\nolimits_{x}(A)}\right]  \right\vert \leq\varepsilon
\cdot\left\Vert A\right\Vert ^{n}%
\end{equation}
where we define the vector $x\in\{-1,1\}$ by $x_{i}=(-1)^{e_{i}}$.
\end{theorem}

\begin{proof}
Following equation (\ref{expected-gly-expansion}) we get
\begin{equation}
\operatorname*{E}_{e\sim\mathcal{D}}\left[  {\operatorname*{Gly}%
\nolimits_{x}(A)}\right]  =\sum_{\sigma_{1},\dotsc,\sigma_{n}\in\lbrack
n]}a_{1,\sigma_{1}}\dotsm a_{n,\sigma_{n}}\operatorname*{E}_{x\sim\mathcal{D}%
}\left[  {(x_{1}\dotsm x_{n})(x_{\sigma_{1}}\dotsm x_{\sigma_{n}})}\right]  .
\end{equation}
If the map $i\mapsto\sigma_{i}$ is a permutation, then $\operatorname*{E}%
_{e\sim\mathcal{D}}\left[  {(x_{1}\dotsm x_{n})(x_{\sigma_{1}}\dotsm
x_{\sigma_{n}})}\right]  $ will always be $1$. \ We would like to bound the
terms for which $\sigma$ does not give a permutation. \ Given such a $\sigma$,
define a vector $c\in\{0,1\}^{n}$ such that
\begin{equation}
c_{i}=1+\left(  \left\vert \left\{  j:\sigma_{j}=i\right\}  \right\vert
\right)  \operatorname*{mod}2.
\end{equation}
Clearly, $c$ is nonzero if $\sigma$ does not give a permutation. \ Then
\begin{align}
\left\vert \operatorname*{E}_{e\sim\mathcal{D}}\left[  {(x_{1}\dotsm
x_{n})(x_{\sigma_{1}}\dotsm x_{\sigma_{n}})}\right]  \right\vert  &
=\left\vert \operatorname*{E}_{e\sim\mathcal{D}}\left[  {(-1)^{c\cdot e}%
}\right]  \right\vert \\
&  \leq\varepsilon
\end{align}
since $\mathcal{D}$ is $\varepsilon$-biased by assumption. \ Then,
\begin{align}
\left\vert \operatorname*{Per}(A)-\operatorname*{E}_{e\sim\mathcal{D}}\left[
{\operatorname*{Gly}\nolimits_{x}(A)}\right]  \right\vert  &  =\left\vert
\sum_{%
\genfrac{}{}{0pt}{}{\sigma_{1},\dotsc,\sigma_{n}\in\lbrack n]}{i\mapsto
\sigma_{i}\text{ not a permutation}}%
}a_{1,\sigma_{1}}\dotsm a_{n,\sigma_{n}}\operatorname*{E}_{e\sim\mathcal{D}%
}\left[  {(x_{1}\dotsm x_{n})(x_{\sigma_{1}}\dotsm x_{\sigma_{n}})}\right]
\right\vert \\
&  \leq\sum_{%
\genfrac{}{}{0pt}{}{\sigma_{1},\dotsc,\sigma_{n}\in\lbrack n]}{i\mapsto
\sigma_{i}\text{ not a permutation}}%
}a_{1,\sigma_{1}}\dotsm a_{n,\sigma_{n}}\left\vert \operatorname*{E}%
_{e\sim\mathcal{D}}\left[  {(x_{1}\dotsm x_{n})(x_{\sigma_{1}}\dotsm
x_{\sigma_{n}})}\right]  \right\vert \label{eqn-use-nonneg}\\
&  \leq\varepsilon\cdot\sum_{%
\genfrac{}{}{0pt}{}{\sigma_{1},\dotsc,\sigma_{n}\in\lbrack n]}{i\mapsto
\sigma_{i}\text{ not a permutation}}%
}a_{1,\sigma_{1}}\dotsm a_{n,\sigma_{n}}\\
&  \leq\varepsilon\cdot\sum_{\sigma_{1},\dotsc,\sigma_{n}\in\lbrack
n]}a_{1,\sigma_{1}}\dotsm a_{n,\sigma_{n}}\\
&  =\varepsilon\cdot\prod_{i=1}^{n}(a_{1,i}+\dotsb+a_{n,i})\\
&  =\varepsilon\cdot\operatorname*{Gly}\nolimits_{(1,\dotsc,1)}(A)\\
&  \leq\varepsilon\cdot\left\Vert A\right\Vert ^{n} \label{eqn-use-gly-bound}%
\end{align}
where for (\ref{eqn-use-nonneg}) we use the fact that the $a_{i,j}$ are
nonnegative, and for (\ref{eqn-use-gly-bound}) we appeal to Proposition
\ref{gly-bound-prop}.
\end{proof}

\begin{corollary}
\label{derandomize-gurvits-corollary}There exists an algorithm to
deterministically approximate the permanent of an $n\times n$ matrix $A$ with
nonnegative entries to within error $\varepsilon\cdot\left\Vert A\right\Vert
^{n}$ which runs in time $\operatorname*{poly}\left(  n,1/\varepsilon\right)
$.
\end{corollary}

\begin{proof}
By Theorem \ref{eps-bias-thm} there is a polynomial time algorithm which
outputs an $\varepsilon$-biased distribution $\mathcal{D}$ on $\{0,1\}^{n}$.
\ Since the seed for this algorithm is only $O\left(  \log n+\log
1/\varepsilon\right)  $ bits, we can compute $\operatorname*{E}_{e\sim
\mathcal{D}}\left[  {\operatorname*{Gly}_{x}(A)}\right]  $ by looping over all
possible seeds. \ By Theorem \ref{eps-bias-implies-approx-thm}, this suffices
as our desired approximation.
\end{proof}

\subsection{Derandomizing the Multicolumn Case\label{DEMULT}}

To derandomize the algorithm of Section \ref{GENGLYNN}, we need to generalize
the concept of $\varepsilon$-bias to arbitrary roots of unity.

\begin{definition}
Given $s_{1},\dotsc,s_{k}$ with $n=s_{1}+\dotsb+s_{k}$, recall that we let
$\mathcal{X}:=\mathcal{R}[s_{1}+1]\times\dotsb\times\mathcal{R}[s_{k}+1]$.
\ Then we say that a probability distribution $\mathcal{D}$ on $\mathcal{X}$
is \textbf{complex-}$\varepsilon$\textbf{-biased} if for any $e_{1}%
,\dotsc,e_{k}$, with $e_{i}\in\{0,\dotsc,s_{i}\}$ where not all $e_{i}$ are
$0$, we have
\begin{equation}
\left\vert \operatorname*{E}_{x\sim\mathcal{D}}\left[  {x_{1}^{e_{i}}\dotsm
x_{k}^{e_{k}}}\right]  \right\vert \leq\varepsilon.
\end{equation}
Also, we say that a probabilistic algorithm $G$ is complex-$\varepsilon
$-biased if it outputs a complex-$\varepsilon$-biased distribution.
\end{definition}

In Section \ref{CEPSBIAS} we will show the following:

\begin{theorem}
\label{complex-eps-bias-thm}There exists a complex-$\varepsilon$-biased
generator which runs in $\operatorname*{poly}\left(  n,1/\varepsilon\right)  $
time and uses $O\left(  \log n+\log1/\varepsilon\right)  $ bits of randomness.
\end{theorem}

These complex-$\varepsilon$-biased distributions are useful because of the
following property, which is analogous to Theorem
\ref{eps-bias-implies-approx-thm}:

\begin{theorem}
\label{complex-eps-bias-implies-approx-thm}If $\mathcal{D}$ is a
complex-$\varepsilon$-biased distribution on $\mathcal{X}$, then for any
$n\times k$ matrix $B$ with nonnegative real entries, let $A$ be an $n\times
n$ matrix with $s_{i}$ copies of the $i^{th}$ column of $B$. \ Then we have
\begin{equation}
\left\vert \operatorname*{Per}(A)-\operatorname*{E}_{x\sim\mathcal{D}}\left[
{\operatorname*{GenGly}\nolimits_{x}(A)}\right]  \right\vert \leq
\varepsilon\cdot\frac{s_{1}!\dotsm s_{k}!}{\sqrt{s_{1}^{s_{1}}\dotsm
s_{k}^{s_{k}}}}\left\Vert B\right\Vert ^{n}.
\end{equation}

\end{theorem}

\begin{proof}
We begin by expanding the product:
\begin{align}
\operatorname*{E}_{x\sim\mathcal{D}}\left[  {\operatorname*{GenGly}%
\nolimits_{x}(A)}\right]   &  =\operatorname*{E}_{x\sim\mathcal{D}}\left[
{\frac{s_{1}!\dotsm s_{k}!}{s_{1}^{s_{1}}\dotsm s_{k}^{s_{k}}}\overline{y_{1}%
}^{s_{1}}\dotsm\overline{y_{k}}^{s_{k}}\prod_{i=1}^{n}\left(  y_{1}%
b_{i,1}+\dotsb+y_{k}b_{i,k}\right)  }\right] \\
&  =\frac{s_{1}!\dotsm s_{k}!}{s_{1}^{s_{1}}\dotsm s_{k}^{s_{k}}}\sum
_{\sigma_{1},\dotsc,\sigma_{n}\in\lbrack k]}b_{1,\sigma_{1}}\dotsm
b_{n,\sigma_{n}}\operatorname*{E}_{x\sim\mathcal{D}}\left[  {\left(
\overline{y_{1}}^{s_{1}}\dotsm\overline{y_{k}}^{s_{k}}\right)  \left(
y_{\sigma_{1}}\dotsm y_{\sigma_{n}}\right)  }\right] \\
&  =s_{1}!\dotsm s_{k}!\sum_{\sigma_{1},\dotsc,\sigma_{n}\in\lbrack
k]}b_{1,\sigma_{1}}\dotsm b_{n,\sigma_{n}}\sqrt{\frac{s_{\sigma_{1}}\dotsm
s_{\sigma_{n}}}{s_{1}^{s_{1}}\dotsm s_{k}^{s_{k}}}}\operatorname*{E}%
_{x\sim\mathcal{D}}\left[  {\left(  \overline{x_{1}}^{s_{1}}\dotsm
\overline{x_{k}}^{s_{k}}\right)  \left(  x_{\sigma_{1}}\dotsm x_{\sigma_{n}%
}\right)  }\right]  .
\end{align}
Say that a sequence $\sigma=(\sigma_{1},\dotsc,\sigma_{n})$ is a
\textit{desired sequence} if for all $i\in\{1,\dotsc,k\}$, we have $\left\vert
\{j:\sigma_{j}=i\}\right\vert =s_{i}$. \ Let
\begin{equation}
e_{i}=\left(  -s_{i}+\left\vert \{j:\sigma_{j}=i\}\right\vert \right)
\operatorname{mod}\left(  {s_{k}+1}\right)  .
\end{equation}
Then we have
\begin{equation}
\left(  \overline{x_{1}}^{s_{1}}\dotsm\overline{x_{k}}^{s_{k}}\right)  \left(
x_{\sigma_{1}}\dotsm x_{\sigma_{n}}\right)  =x_{1}^{e_{1}}\dotsm x_{n}^{e_{n}%
}. \label{eqn-simplify-product-to-e}%
\end{equation}
Note that $e_{i}=0$ for all $i$ if and only if $\sigma$ is desired. \ Thus, in
the case that $\sigma$ is desired, the expected value of
(\ref{eqn-simplify-product-to-e}) is $1$. \ Otherwise, since $\mathcal{D}$ is
complex-$\varepsilon$-biased, we get
\[
\left\vert \operatorname*{E}_{x\sim\mathcal{D}}\left[  {\left(  \overline
{x_{1}}^{s_{1}}\dotsm\overline{x_{k}}^{s_{k}}\right)  \left(  x_{\sigma_{1}%
}\dotsm x_{\sigma_{n}}\right)  }\right]  \right\vert \leq\varepsilon.
\]
Therefore,
\begin{align}
&  \left\vert \operatorname*{Per}(A)-\operatorname*{E}_{x\sim\mathcal{D}%
}\left[  {\operatorname*{GenGly}\nolimits_{x}(A)}\right]  \right\vert \\
&  =s_{1}!\dotsm s_{k}!\left\vert \sum_{%
\genfrac{}{}{0pt}{}{\sigma_{1},\dotsc,\sigma_{n}\in\lbrack k]}{\sigma\text{
not desired}}%
}\left(  b_{1,\sigma_{1}}\dotsm b_{n,\sigma_{n}}\right)  \sqrt{\frac
{s_{\sigma_{1}}\dotsm s_{\sigma_{n}}}{s_{1}^{s_{k}}\dotsm s_{k}^{s_{k}}}%
}\operatorname*{E}_{x\sim\mathcal{D}}\left[  {\left(  \overline{x_{1}}^{s_{1}%
}\dotsm\overline{x_{k}}^{s_{k}}\right)  \left(  x_{\sigma_{1}}\dotsm
x_{\sigma_{n}}\right)  }\right]  \right\vert \\
&  \leq\frac{s_{1}!\dotsm s_{k}!}{\sqrt{s_{1}^{s_{1}}\dotsm s_{k}^{s_{k}}}%
}\sum_{%
\genfrac{}{}{0pt}{}{\sigma_{1},\dotsc,\sigma_{n}\in\lbrack k]}{\sigma\text{
not desired}}%
}\left(  b_{1,\sigma_{1}}\dotsm b_{n,\sigma_{n}}\right)  \sqrt{s_{\sigma_{1}%
}\dotsm s_{\sigma_{n}}}\left\vert \operatorname*{E}_{x\sim\mathcal{D}}\left[
{\left(  \overline{x_{1}}^{s_{1}}\dotsm\overline{x_{k}}^{s_{k}}\right)
\left(  x_{\sigma_{1}}\dotsm x_{\sigma_{n}}\right)  }\right]  \right\vert \\
&  \leq\varepsilon\cdot\frac{s_{1}!\dotsm s_{k}!}{\sqrt{s_{1}^{s_{1}}\dotsm
s_{k}^{s_{k}}}}\sum_{%
\genfrac{}{}{0pt}{}{\sigma_{1},\dotsc,\sigma_{n}\in\lbrack k]}{\sigma\text{
not desired}}%
}\left(  b_{1,\sigma_{1}}\dotsm b_{n,\sigma_{n}}\right)  \sqrt{s_{\sigma_{1}%
}\dotsm s_{\sigma_{n}}}\\
&  \leq\varepsilon\cdot\frac{s_{1}!\dotsm s_{k}!}{\sqrt{s_{1}^{s_{1}}\dotsm
s_{k}^{s_{k}}}}\sum_{\sigma_{1},\dotsc,\sigma_{n}\in\lbrack k]}\left(
b_{1,\sigma_{1}}\dotsm b_{n,\sigma_{n}}\right)  \sqrt{s_{\sigma_{1}}\dotsm
s_{\sigma_{n}}}\\
&  =\varepsilon\cdot\frac{s_{1}!\dotsm s_{k}!}{\sqrt{s_{1}^{s_{1}}\dotsm
s_{k}^{s_{k}}}}\prod_{i=1}^{n}\left(  \sqrt{s_{1}}b_{i,1}+\dotsb+\sqrt{s_{k}%
}b_{i,k}\right) \\
&  =\varepsilon\cdot\operatorname*{GenGly}\nolimits_{(1,\dotsc,1)}(A)\\
&  \leq\varepsilon\cdot\frac{s_{1}!\dotsm s_{k}!}{\sqrt{s_{1}^{s_{1}}\dotsm
s_{k}^{s_{k}}}}\left\Vert B\right\Vert ^{n}%
\end{align}
where the last inequality follows from Lemma \ref{gengly-bound-lemma}.
\end{proof}

Theorems \ref{complex-eps-bias-thm} and
\ref{complex-eps-bias-implies-approx-thm} immediately imply the following.

\begin{corollary}
For any $n\times k$ matrix $B$ with nonnegative real entries, let $A$ be an
$n\times n$ matrix with $s_{i}$ copies of the $i^{th}$\ column of $B$. \ Then
we can deterministically approximate $\operatorname*{Per}(A)$ to within
additive error
\begin{equation}
\varepsilon\cdot\frac{s_{1}!\dotsm s_{k}!}{\sqrt{s_{1}^{s_{1}}\dotsm
s_{k}^{s_{k}}}}\left\Vert B\right\Vert ^{n}%
\end{equation}
in $\operatorname*{poly}\left(  n,1/\varepsilon\right)  $ time.
\end{corollary}

\subsection{Constructing Complex-$\varepsilon$-Biased
Generators\label{CEPSBIAS}}

Our proof of Theorem \ref{complex-eps-bias-thm} will be modeled after the
proof of Theorem \ref{eps-bias-thm} given by Naor and Naor \cite{naor-naor}.
\ We just need to generalize each step to deal with the domain of complex
roots of unity rather than the binary domain $\left\{  1,-1\right\}  $.

Recall that our goal is to generate, with as few bits of randomness as
possible, complex numbers $(\zeta_{1},\dotsc,\zeta_{k})\in\mathcal{X}$, such
that for any exponents $e_{1},\dotsc,e_{k}$, we have
\begin{equation}
\left\vert \operatorname*{E}\left[  {\zeta_{1}^{e_{1}}\dotsm\zeta_{k}^{e_{k}}%
}\right]  \right\vert \leq\varepsilon.
\end{equation}
In our construction, it will be simpler to adopt an alternate perspective,
which we now describe. \ If we let $f_{j}$ be such that $\zeta_{j}=e^{2\pi
if_{j}/(s_{j}+1)}$, and let $\xi_{j}=e^{2\pi ie_{j}/(s_{j}+1)}$, then
\begin{equation}
\zeta_{j}^{e_{j}}=\left(  e^{2\pi i/(s_{j}+1)}\right)  ^{e_{j}f_{j}}=\xi
_{j}^{f_{j}}.
\end{equation}
Hence our task is equivalent to outputting a random $k$-tuple $(f_{1}%
,\dotsc,f_{k})$ such that for any choice of complex numbers $\xi_{j}=e^{2\pi
ie_{j}/(s_{j}+1)}$, not all $1$, we have
\begin{equation}
\left\vert \operatorname*{E}\left[  {\xi_{1}^{f_{1}}\dotsm\xi_{k}^{f_{k}}%
}\right]  \right\vert \leq\varepsilon.
\end{equation}
The proof of Theorem \ref{eps-bias-thm} given by Naor and Naor
\cite{naor-naor} revolves around the ability to find vectors $r\in\{0,1\}^{n}$
such that $r\cdot x=1\operatorname*{mod}2$ with constant probability for any
nonzero $x\in\{0,1\}^{n}$. \ A uniformly random $r\in\{0,1\}^{n}$ will satisfy
$r\cdot x=1\operatorname*{mod}2$ with probability $1/2$, and in fact it is
also possible to generate such an $r$ using only a small amount of randomness
(possibly allowing for a weaker constant).

In our setting, this would translate to generating $(f_{1},\dotsc,f_{k})$ such
that $\xi_{1}^{f_{1}}\dotsm\xi_{k}^{f_{k}}=-1$ with constant probability.
\ Unfortunately, this will not be possible since the $\xi_{i}$ do not come
from the binary domain $\left\{  1,-1\right\}  $. \ Instead, we will have to
replace the concept of \textquotedblleft being equal to $-1$\textquotedblright%
\ with what we call \textquotedblleft$\theta$-strongness.\textquotedblright

\begin{definition}
For a unit-norm complex number $\lambda$, we say that $\lambda$ is $\theta
$\textbf{-strong} if $\left\vert \arg\lambda\right\vert \geq\theta$, where
$\arg\lambda$ is the phase $\phi\in\lbrack-\pi,\pi)$ such that $\lambda
=e^{i\phi}$.
\end{definition}

\begin{figure}[ptb]
\begin{center}
\includegraphics[scale=0.8]{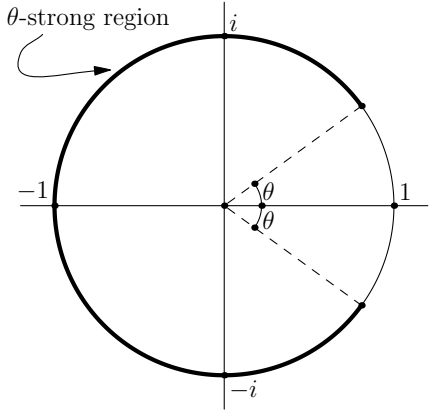}
\end{center}
\caption{The $\theta$-strong region of the unit circle in the complex plane.}%
\end{figure}

\begin{lemma}
\label{construct-strong-lemma}We can generate $f_{1},\dotsc,f_{k}$ such that
for any $\xi_{1},\dotsc,\xi_{k}$, not all $1$, we have
\begin{equation}
\Pr[\xi_{1}^{f_{1}}\dotsm\xi_{k}^{f_{k}}\text{ is }\pi/8\text{-strong}%
]\geq\frac{1}{16}%
\end{equation}
using $O(\log n)$ bits of randomness.
\end{lemma}

The proof of this will require the following fact.

\begin{proposition}
\label{pairwise-independence-prop}We can generate $k$ random variables
$(g_{1},\dotsc,g_{k})$ such that we have $g_{i}\in\{0,\dotsc,s_{i}\}$ (nearly)
uniformly, and $c$-wise independence between the $g_{i}$, using $O(c\log n)$
bits of randomness.
\end{proposition}

This is well-known. \ See, e.g., Luby and Wigderson \cite{luby-wigderson}.
\ It is possible to generate $k$ random variables $t_{1},\dotsc,t_{k}$ in
$\mathbb{F}_{p}$, where $p$ is a prime of size $\operatorname*{poly}(k)$, such
that each one is uniformly distributed in $\mathbb{F}_{p}$, and any $c$ of
them are independent, using only $c\log k$ random bits. \ Since our desired
domains are $\{0,\dotsc,s_{i}\}$, to get the above result, we have to take the
$t_{i}$ mod $s_{i} + 1$. \ There is a small error since $p$ is not divisible
by $s_{i} + 1$, but this error is small enough that it does not matter for our purposes.

Another fact we will need was proved by Naor and Naor \cite{naor-naor}:

\begin{proposition}
\label{pairwise-small-intersection-prop}There is a constant $c$ such that the
following holds: Let $\ell$ and $m$ be integers such that $m\leq\ell\leq2m$,
and let $S$ is a fixed subset of $\left[  k\right]  $ with $\left\vert
S\right\vert =\ell$. \ If $T$ is a randomly chosen subset of $\left[
k\right]  $ such that each element $i\in\left[  k\right]  $ is in $T$ with
probability $2/m$, and the elements are chosen with pairwise independence,
then
\begin{equation}
\Pr\left[  \left\vert S\cup T\right\vert \leq c\right]  \leq\frac{1}{2}.
\end{equation}

\end{proposition}

In particular, they show that we can take $c=7$. \ These facts are enough for
the following proof:

\begin{proof}
[Proof of Lemma \ref{construct-strong-lemma}]Suppose that $\ell\geq1$ of the
$\xi_{i}$ are not equal to $1$, and let $h$ be the integer such that
$2^{h}\leq\ell<2^{h+1}$. \ For now, suppose that we know $h$; we will relax
this assumption later. \ Do the following (using Proposition
\ref{pairwise-independence-prop}):
\begin{itemize}
\item Generate $u = (u_{1},\dotsc,u_{k})$ whose entries are $c$-wise
independent and such that $u_{i}$ is (nearly) uniformly distributed in
$\{0,\dotsc,s_{i}\}$, in the manner stated above. \ Here, $c$ is the constant
from Proposition \ref{pairwise-small-intersection-prop}.
\item Generate $w=(w_{1},\dotsc,w_{k})$ whose entries are pairwise independent
and such that $w_{i}\in\left\{  0,1\right\}  $ and $\Pr[w_{i}=1]=\min
(1/2^{h-1},1)$.
\end{itemize}
Finally, we define $f=(f_{1},\dotsc,f_{k})$ by
\begin{equation}
f_{i}=\left\{
\begin{tabular}
[c]{ll}%
$0$ & if $w_{i}=0$\\
$u_{i}$ & if $w_{i}=1$%
\end{tabular}
\ \right.
\end{equation}
Let $V:=\{i:\xi_{i}\neq1\}$ and $W:=\{i:w_{i}=1\}$. \ Now we need that with
probability at least $1/2$, we will have $1\leq\left\vert V\cap W\right\vert
\leq c$. \ This follows from Proposition
\ref{pairwise-small-intersection-prop}. Now let us examine $\lambda:=\xi
_{1}^{f_{1}}\dotsm\xi_{k}^{f_{k}}$. \ With probability $1/2$, we have
$1\leq\left\vert V\cap W\right\vert \leq c$. \ If this is the case, write
$\lambda=\xi_{i_{1}}^{f_{i_{1}}}\dotsm\xi_{i_{d}}^{f_{i_{d}}}$ where $V\cap
W=\left\{  i_{1},\ldots,i_{d}\right\}  $, and $1\leq d\leq c$. \ Then, the
values of $f_{i_{j}}$ will be (nearly) uniform in $\{0,\dotsc,s_{i_{j}}\}$ and
independent. \ Because of the $c$-independence, for any fixed values of
$f_{i_{2}},\dotsc,f_{i_{d}}$, we can condition on those choices and still get
that $f_{i_{1}}$ is uniform in $\{0,\dotsc,s_{i_{j}}\}$. \ Now, we know that
$\xi_{i_{1}}$ is an $s_{i_{1}}^{\text{th}}$ root of unity other than $1$ (but
still not necessarily primitive). \ Therefore as $\xi_{i_{1}}^{f_{i_{1}}}$
ranges around the unit circle, we can see that
\begin{equation}
\xi_{i_{1}}^{f_{i_{1}}}\left(  \xi_{i_{1}}^{f_{i_{1}}}\dotsm\xi_{i_{d}%
}^{f_{i_{d}}}\right)  =\xi_{i_{1}}^{f_{i_{1}}}\left(  \text{fixed unit-norm
complex number}\right)
\end{equation}
will have a probability of at least $1/4$ of being $\pi/4$-strong (being
conservative with constants here). \ Putting it all together, we see that
$\lambda$ has a $1/8$ probability of being $\pi/4$-strong, and we have
generated $\lambda$ using only $O(\log n)$ bits. However, this has all assumed
that we know the value of $h$. \ Suppose we do not know $h$. \ Then, generate
$O(\log n)$ random bits, and for every $h\in\{0,\dotsc,\left\lfloor \log
_{2}k\right\rfloor \}$ use these bits to construct $\lambda^{(h)}$ using the
above method. \ Let $h_{\text{actual}}$ be the value such that
$2^{h_{\text{actual}}}\leq\ell<2^{h_{\text{actual}}\ +1}$. Now, randomly and
indepently choose bits $b_{0},\dotsc,b_{\log_{2}k}$ and let $\lambda
:=\prod_{h}\left(  \lambda^{(h)}\right)  ^{b_{h}}$. \ There is a $1/8$ chance
that $\lambda^{(h_{\text{actual}})}$ is $\pi/4$-strong. \ If this happens to
be the case, then write
\begin{equation}
\lambda=\left(  \lambda^{(h_{\text{actual}})}\right)  ^{b_{h_{\text{actual}}}%
}\prod_{h\neq h_{\text{actual}}}\left(  \lambda^{(h)}\right)  ^{b_{h}}%
\end{equation}
If we fix the values of $b_{h}$ for $h\neq h_{\text{actual}}$, then the factor
on the right becomes fixed, call it $\Lambda$. \ If $\Lambda$ is $\pi
/8$-strong, then $\lambda$ will be $\pi/8$-strong if $b_{h_{\text{actual}}}%
=0$. \ If $\Lambda$ is not $\pi/8$-strong, then $\lambda$ will be $\pi
/8$-strong if $b_{h_{\text{actual}}}=1$ (here, using the fact that
$\lambda^{(h_{\text{actual}})}$ is $\pi/4$-strong). \ In either case, there is
a $1/2$ chance that $\lambda$ is $\pi/8$-strong. \ Combining with the $1/8$
probabililty from earlier, we get that in total, there is at least a $1/16$
chance that $\lambda$ is $\pi/8$-strong.
\end{proof}

However, we need a better guarantee than just constant probability given by
Lemma \ref{construct-strong-lemma}. \ In particular, we will need to construct
a set of $\lambda$ which has high probability (at least $1-\varepsilon/2$) of
containing many $\pi/8$-strong elements. \ We could sample $O(\log
(1/\varepsilon))$ samples independently, but that would require too many
random bits.

To reduce the number of random bits, we use \textit{deterministic
amplification}. \ The following result was shown independently by Cohen and
Wigderson \cite{cohen-wigderson} and by Impagliazzo and Zuckerman
\cite{impagliazzo-zuckerman}, using ideas from Ajtai, Koml\'{o}s, and
Szemer\'{e}di \cite{aks2}. \ The original motivation for these results was to
amplify error probabilities in \textsf{BPP} algorithms using few bits of
randomness. \ Here is a formulation which is useful to us:

\begin{theorem}
\label{deterministic-amplification-thm}Let $\mathcal{T}=\left\{  0,1\right\}
^{r}$. \ Then with $r+O\left(  \ell\right)  $ bits of randomness we can
efficiently generate a set $y_{1},\dotsc,y_{\ell}\in\mathcal{T}$ such that for
any $\mathcal{S}\subseteq\mathcal{T}$ with $\left\vert \mathcal{S}\right\vert
\geq2\left\vert \mathcal{T}\right\vert /3$, we have
\begin{equation}
\Pr[y_{i}\in\mathcal{S}\text{ for at least $\ell/2$ values of }i]\geq
1-2^{-\Omega\left(  \ell\right)  }.
\end{equation}

\end{theorem}

This gives the following lemma:

\begin{lemma}
\label{construct-strong-highprob-lemma}There are constants $p,q>0$ such that
using $O\left(  \log n+\ell\right)  $ bits of randomness, we can generate
$\lambda_{1},\dotsc,\lambda_{\ell}$ such that with probability $1-2^{q\ell}$,
there are at least $p\ell$ values of $i$ such that $\lambda_{i}$ is $\pi/8$-strong.
\end{lemma}

\begin{proof}
We start with a small amount of initial amplification: let $t$ be the constant
such that if we generate $\lambda_{1},\dotsc,\lambda_{t}$ independently as in
Lemma \ref{construct-strong-lemma}, then with probability at least $2/3$, one
of the $\lambda_{i}$ is $\pi/8$-strong, and let $r$ be the number of random
bits needed to generate these $\lambda_{1},\dotsc,\lambda_{t}$. \ If we want
to generate $\lambda_{1},\dotsc,\lambda_{\ell}$ for some $\ell$, then generate
groups $\lambda_{j,1},\dotsc,\lambda_{j,t}$ for $j\in\left[  \ell/t\right]  $
according to Theorem \ref{deterministic-amplification-thm}. \ Then with
probability at least $1-2^{-\Omega\left(  \ell/t\right)  }=1-2^{-\Omega\left(
\ell\right)  }$, at least half of the groups will have at least one $\pi
/8$-strong member, i.e., at least a $1/2t$ fraction of the $\lambda$'s will be
$\pi/8$-strong. \ The total number of random bits used is $r+O\left(
\ell\right)  =O\left(  \log n+\ell\right)  $.
\end{proof}

It is worth noting that Lemma \ref{construct-strong-highprob-lemma} is a
slight difference between our construction of complex-$\varepsilon$-biased
distributions and the construction of $\varepsilon$-biased distributions by
Naor and Naor \cite{naor-naor}. \ Naor and Naor constructed many values,
needing only one to be $-1$, while we construct many values needing
\textit{many} of them to be $\pi/8$-strong, for reasons evident in the proof
of Theorem \ref{complex-eps-bias-thm}. \ Luckily, this was not a problem,
since as we saw, the same class of deterministic amplification results
applied, allowing us to prove Lemma \ref{construct-strong-highprob-lemma}.

\begin{proof}
[Proof of Theorem \ref{complex-eps-bias-thm}]Generate $\lambda_{1}%
,\dotsc,\lambda_{\ell}$ as in Lemma \ref{construct-strong-highprob-lemma},
choosing
\begin{equation}
\ell=\left\lceil \max\left(  \frac{\log_{1/2}\left(  \varepsilon/2\right)
}{q},\frac{\log_{\beta}\left(  \varepsilon/2\right)  }{p}\right)  \right\rceil
.
\end{equation}
where $\beta:=\frac{1}{2}\left\vert 1+e^{\pi i/8}\right\vert $. \ Randomly and
independently choose $d_{1},\dotsc,d_{\ell}$, each uniformly from $\{0,1\}$.
\ Then let
\begin{equation}
\mu:=\prod_{j=1}^{\ell}\lambda_{i}^{d_{i}}.
\end{equation}
With probability at least $1-2^{-q\ell}\geq1-\varepsilon/2$, at least $p\ell$
of the $\lambda_{j}$ will be $\pi/8$-strong. \ Suppose that this is the case.
\ Then
\begin{align}
\left\vert \operatorname*{E}_{d\sim\{0,1\}^{\ell}}\left[  \mu\right]
\right\vert  &  =\prod_{j=1}^{\ell}\left\vert \frac{1+\lambda_{j}}%
{2}\right\vert \\
&  \leq\beta^{p\ell}\\
&  \leq\varepsilon/2.
\end{align}
Thus, this shows that we get that with probability at least $1-\varepsilon/2$,
the norm of the expected value of $\mu$ is at most $\varepsilon/2$. \ Hence,
the total expected value of $\mu$, over the random choices of $\lambda
_{1},\dotsc,\lambda_{\ell}$ and $d_{1},\dotsc,d_{\ell}$ is at most
\begin{equation}
\left(  1-\frac{\varepsilon}{2}\right)  \frac{\varepsilon}{2}+\frac
{\varepsilon}{2}\left(  1\right)  \leq\varepsilon.
\end{equation}
Now, since we have to choose $\ell$ values $d_{j}$, each in $\left\{
0,1\right\}  $, this takes $O\left(  \ell\right)  =O\left(  \log
1/\varepsilon\right)  $ bits of randomness. \ Generating the $\lambda
_{1},\dotsc,\lambda_{\ell}$ takes $O\left(  \log n+\ell\right)  =O\left(  \log
n+\log1/\varepsilon\right)  $ bits of randomness. \ Therefore, the entire
procedure then takes $O\left(  \log n+\log1/\varepsilon\right)  $ bits of randomness.
\end{proof}

\section{Open Problems\label{OPEN}}

There are many interesting open problems about Gurvits's algorithm, and the
classical simulability of quantum optics more generally. \ Firstly, can we
improve the error bound even further in the case that the matrix $A$ has both
repeated rows \textit{and} repeated columns? \ More specifically, can we
estimate any linear-optical amplitude%
\begin{equation}
\left\langle s_{1},\dotsc,s_{k}|\varphi\left(  U\right)  |t_{1},\dotsc
,t_{k}\right\rangle =\frac{\operatorname*{Per}\left(  U_{s_{1},\dotsc
,s_{k},t_{1},\dotsc,t_{k}}\right)  }{\sqrt{s_{1}!\cdots s_{k}!t_{1}!\cdots
t_{k}!}}%
\end{equation}
to $\pm1/\operatorname*{poly}\left(  n\right)  $\ additive error (or
better)\ in polynomial time? \ This is related to a question raised by
Aaronson and Arkhipov \cite{aark}, of whether quantum optics can be used to
solve any \textit{decision} problems that are classically intractable (rather
than just sampling and search problems).

An even more obvious problem is to generalize our derandomization of Gurvits's
algorithm so that it works for \textit{arbitrary} matrices, not only matrices
with nonnegative entries. \ Unfortunately, our current technique appears
unable to handle arbitrary matrices, as it relies too heavily on the
expansions of $\operatorname*{Gly}_{(1,\dotsc,1)}\left(  A\right)  $ and
$\operatorname*{GenGly}_{(1,\dotsc,1)}\left(  A\right)  $ having no negative terms.

An alternative approach to derandomizing Gurvits's algorithm would be to prove
a general \textquotedblleft structure theorem,\textquotedblright\ showing that
if $U$ is a unitary or subunitary matrix, then $\operatorname*{Per}\left(
U\right)  $\ can only be non-negligibly large if $U$ has some special form:
for example, if $U$ is close (in a suitable sense) to a product of permutation
and diagonal matrices. \ A deterministic algorithm to estimate
$\operatorname*{Per}\left(  U\right)  $\ could then follow, by simply checking
whether the structure is present or not. \ We conjecture that such a structure
theorem holds, but do not even know how to formulate the conjecture precisely.

Even in the realm of nonnegative entries, there is much room for improvement.
\ Among other things, it would be extremely interesting to derandomize a
Jerrum-Sinclair-Vigoda-type algorithm \cite{jsv}, enabling us
deterministically to approximate the permanent of nonnegative matrices to
within small \textit{multiplicative} error.

\section{Acknowledgments}

We thank Alex Arkhipov, Michael Forbes, Leonid Gurvits, and Shravas Rao for
helpful discussions.

\bibliographystyle{plain}
\bibliography{thesis}

\end{document}